\documentclass[11pt]{amsart}

\usepackage{amsfonts}
\usepackage{amssymb, amscd}
\usepackage{latexsym}

\catcode`\@=11

 \def\AMSTeXfeatures{\Plainheads 
   \let\current@vert=\AMS@vert}

 \def\Plainheads{\sh@ftdiam=0.05em
   \getlabeldims
   \let\vshaftfill=\plnvsolidfill
   \let\hshaftfill=\plnhsolidfill
   \let\th@rhead=\plnrhead
   \let\th@lhead=\plnlhead
   \let\th@dnhead=\plndnhead
   \let\th@uphead=\plnuphead}
 
 \def\glet{\global\let}

 \def\LaTeXfeatures{\catcode`\@=11
   \ifx\@clnwd\undefined \nol@g
      \input ltxcode.tex \dol@g \fi
   \ltxheads \let\current@vert=\new@vert
   \providelto \catcode`\@=\active}

 \def\nol@g{\def\wlog{\edef\garbage}}
 \def\dol@g{\let\wlog=\wl@g} \let\wl@g=\wlog
 \nol@g 

 \newbox\ltobox
 \def\providelto{{\setbox\z@=
   \hbox{$\to$}\minharrlen=\wd\z@
   \global\setbox\ltobox=\hbox{$\activeat>>>$}}
   \def\lto{\mathrel{\copy\ltobox}}}

 \def\ltxheads{\sh@ftdiam=\@wholewidth
   \getlabeldims
   \let\vshaftfill= \ltxvsolidfill
   \let\hshaftfill=\ltxhsolidfill
   \let\th@rhead=\ltxrhead
   \let\th@lhead=\ltxlhead
   \let\th@dnhead=\ltxdnhead
   \let\th@uphead=\ltxuphead}
 {\catcode`\@=\active
   \gdef@#1{\csname #1\string@at\endcsname}
   \glet\activeat=@}
 \def\def@#1{\expandafter\def\csname #1@at\endcsname}

 \def@>#1>#2>{\@rrow R{#1}{#2}}
 \def@<#1<#2<{\@rrow L{#1}{#2}}
 \def@ V#1V#2V{\@rrow V{#1}{#2}}
 \def@ A#1A#2A{\@rrow A{#1}{#2}}
 \def@/#1/#2/#3/{\@rrow{#1}{#2}{#3}}
 \def@.{\ifodd\row\ifmmode\noharrow
     \else\leavevmode.\spacefactor3000 \fi
   \else\novarrow\fi}
 \def@={\ifodd\row\harrow\hequalfill{}{}%
   \else\varrow\vequalfill{}{}\fi}
 \def@:#1{\ifx=#1\harrow\deffill{}{}%
   \else\leavevmode\null:#1\fi}
 \def@|{\current@vert}
  \def\AMS@vert{\varrow\vequalfill{}{}}
  \def\new@vert#1|#2|{\ifodd\row
   \let\nextarrow\vertexvarrow
   \else\let\nextarrow\varrow\fi
   \nextarrow\vshaftfill{#1}{#2}}
 \def@-{\ifmmode\let\next\hl@ne
   \else\let\next\AMSatdash \fi \next}
  \def\hl@ne#1-#2-{\harrow\hshaftfill{#1}{#2}}
  \def\AMSatdash{\let\next\relax\leavevmode
    \def\next@{\ifx\next-%
      \def\next-{\futurelet\next\nextii@}%
     \else\def\next{\hbox{-}}\fi\next}%
    \def\nextii@{\ifx\next-\def\next-{\hbox{---}}%
      \else\def\next{\hbox{--}}\fi\next}%
    \futurelet\next\next@}
 \def@(#1){\tweenarrows{#1}}
 \def@[#1]{\setsp@n#1\relax\activeat}
 \def\fiberbox{\hbox{$\vcenter{\hr@le\hbox{\vr@le
   \kern1ex\vbox{\kern1.2ex}\vr@le}\hr@le}$}}
  \def\hr@le{\hrule height \sh@ftdiam}
  \def\vr@le{\vrule width \sh@ftdiam}
 \def@+#1+#2+#3+{\ifodd\row \harrow{#1}{#2}{#3}%
   \else \varrow{#1}{#2}{#3}\fi}

 \def\Dnarrfill{\vequalfill\Dnhe@d}
 \def\Uparrfill{\Uphe@d\vequalfill}
 
 \def\ontofill{\rtarrfill\kern-0.3em 
   \th@rhead\kern 0.3em} 

 \def\rtarrfill{\hshaftfill\th@rhead}
 \def\ltarrfill{\th@lhead\hshaftfill}
 \def\dnarrfill{\vshaftfill\th@dnhead}
 \def\uparrfill{\th@uphead\vshaftfill}
 \def\hequalfill{\plnhfill=}
 \def\deffill{:\plnhfill=}
 \def\plnvextfill#1{\setbox\z@
   \hbox{\the\textfont3 #1}%
   \dimen@=\dp\z@\advance\dimen@\ht\z@
   \copy\z@ \kern-\dimen@ 
   \cleaders\copy\z@ \vfill
   \kern-\dimen@ 
   \box\z@}
 \def\plnhfill#1{$\m@th\mkern-1.5mu\mathord#1\mkern-6mu
    \cleaders\hbox{$\mkern-2mu\mathord#1\mkern-2mu$}\hfill
    \mkern-6mu\mathord#1\mkern-1.5mu$}
 \def\vequalfill{\plnvextfill{\char'167}}
 \def\plnvsolidfill{\plnvextfill{\char'077}}
 \def\plnhsolidfill{\plnhfill-}
 \def\ltxhsolidfill{\leaders\hrule height\topofshaft depth\botofshaft
   \hfill}
 \def\ltxvsolidfill{\leaders\vrule width\sh@ftdiam\vfill}
 \def\hdashfill{\hd@sh\wd@sh
   \xleaders \hbox{\wd@sh\hd@sh\wd@sh}\hfill
   \wd@sh\hd@sh}
 \def\vdashfill{\vd@sh\wd@sh
   \xleaders \vbox{\wd@sh\vd@sh\wd@sh}\vfill
   \wd@sh\vd@sh}
 \def\dashed{\ifinmeasureCD\else
    \ifodd\row\option{\let\hshaftfill=\hdashfill}%
   \else\option{\let\vshaftfill=\vdashfill}\fi\fi}

 \newdimen\CDstrutht  \newdimen\CDstrutdp
 \newdimen\CDstrutlen \CDstrutlen=\CDstrutht
   \advance\CDstrutlen by \CDstrutdp

 \def\CDstrut{\vrule
   height \ifnum\row=1 \z@\else\CDstrutht \fi
   depth \ifnum\row=\numrows \z@ \else\CDstrutdp \fi
   width\z@}

 \newdimen\CDarrsurr \CDarrsurr=0.375em
 \newdimen\CDdashlen
 \newdimen\CDvarrlen \CDvarrlen=1.5\baselineskip
 \newdimen\minharrlen 
  \setbox\z@\hbox{$\longrightarrow$} \minharrlen=\wd\z@
 \newdimen\minCDharrlen \minCDharrlen=2.5em 
\newdimen \minc@lwd
\def\findminc@lwd{\minc@lwd=2\CDarrsurr
  \advance\minc@lwd\minCDharrlen}

 \newdimen\sh@ftdiam

 \newdimen\labelsurr \labelsurr=1.25 em

\newcount\sp@ncnt \sp@ncnt=\@ne
\newcount\sp@ncnt@ \sp@ncnt@=\@ne
\newdimen\@rrwd \newdimen\@rrdp

 \def\adjustbot#1{\option{\advance\@rrdp#1\relax}}
 
\def\pushvertex#1{\global\p@shlen#1\relax
   \global\let\maybepush=\dopush}

 \newdimen\p@shlen \p@shlen=\z@

 \let\maybepush=\relax
 \def\dopush{\ifinmeasureCD 
   \advance\locdimen by -\p@shlen 
   \else\advance \@rrwd by -\p@shlen \fi 
   \global\let\maybepush=\relax \global\p@shlen=\z@\relax}

 \def\span@ne{\global\sp@ncnt=\@ne\relax}
 \def\setsp@n#1#2{\global\sp@ncnt=#1\relax
   \ifx\relax#2\relax\else\global\sp@ncnt@=#2\relax\fi}

 \def\plnrhead{\llap{$\rightarrow\mkern-1.5mu$}}
 \def\plnlhead{\rlap{$\mkern-1.5mu\leftarrow$}}

 \def\clap#1{\hbox to \z@{\hss #1\hss}}

 \def\plndnhead{\hbox{\the\textfont3 \char'171}}
 \def\plnuphead{\hbox{\the\textfont3 \char'170}}
 \def\Dnhe@d{\hbox{\the\textfont3 \char'177}}
 \def\Uphe@d{\hbox{\the\textfont3 \char'176}}

 \def\ltxrhead{\raise\@xisheight
   \llap{\smash{\@linefnt\@getrarrow(1,0)}}}
 \def\ltxlhead{\raise\@xisheight
   \rlap{\@linefnt\@getlarrow(-1,0)}}
 \def\ltxuphead{\setbox\z@=\rlap{%
   \kern\@halfwidth\@linefnt\char'66}%
   \copy\z@\kern-\ht\z@}
 \def\ltxdnhead{\setbox\z@=\rlap{%
   \kern\@halfwidth\@linefnt\char'77}%
   \ht\z@=\z@\box\z@}

 \def\wd@sh{\kern0.5\CDdashlen}
 \def\hd@sh{\vrule height\topofshaft depth\botofshaft
    width\CDdashlen}
 \def\vd@sh{\hrule height\CDdashlen
   depth\z@ width\sh@ftdiam}

\def\xylist{14{3434}13{2414}12{1723}%
  23{1413}34{1153}11{0867}43{0707}%
  32{0580}21{0414}31{0291}41{0}}
\newcount\tgtcnt@
\def\find@xyargs{\dimen@=\@rrdp
  \advance\dimen@ by \CDstrutlen
  \tgtcnt@=\dimen@ \dimen@=\@rrwd 
  \divide\dimen@ by \@m 
  \divide \tgtcnt@ by \dimen@ 
  \expandafter\testxy\xylist\relax
  \unitlength=\@xarg\@rrdp
  \divide\unitlength by\@yarg\relax}
\def\testxy#1#2#3{\ifnum\tgtcnt@>#3
    \@xarg=#1\relax \@yarg=#2\relax
    \let\next=\ignorerest
  \else\let\next\testxy\fi\next}
\def\ignorerest#1\relax{\relax}

\let\scalefactor=\@ne
\def\SWarrow{\find@xyargs\vector
  (-\@xarg,-\@yarg)\scalefactor\hskip-\wd\@linechar}
\def\NWarrow{\find@xyargs\vector
  (-\@xarg,\@yarg)\scalefactor\hskip-\wd\@linechar}
\def\NEarrow{\find@xyargs\vector
  (\@xarg,\@yarg)\scalefactor}
\def\SEarrow{\find@xyargs\vector
  (\@xarg,-\@yarg)\scalefactor}
\def\rightupline{\find@xyargs\@linelen=\scalefactor
     \unitlength\@sline}
\def\rightdownline{\find@xyargs\@yarg=-\@yarg\relax
     \@linelen=\scalefactor\unitlength\@sline}

\def\Sim{\ifodd\row\setbox\z@=\hbox{$\sim$}\dimen@=\ht\z@
 \advance\dimen@ by -\@xisheight
  \vbox{\box\z@\kern-\@xisheight\kern\dimen@}%
  \else\hbox{$\wr$}\fi}

\def\harrow#1#2#3{\inmeasureCDtrue\findminarrwd
  {#2}{#3}{\sp@ncnt\minharrlen}\inmeasureCDfalse\span@ne
  \mathrel{\hbox{\options\hplace{#1}\ulabel{#2}\dlabel{#3}}}}

\def\noharrow{\harrow\hfill{}{}}
\def\vertexvarrow#1#2#3{\findarrdp \@rrwd=\z@ \setsp@n\@ne\@ne
  \vbox to \z@{\kern-1.2\CDstrutht
  \rlap{\options\vplace{#1}\llabel{#2}\rlabel{#3}}\vss}}

\newif\ifinmeasureCD
\def\measurelabel#1{\setbox\z@
  \hbox{$\scriptstyle#1\kern\labelsurr$}%
  \ifdim\wd\z@>\@rrwd \@rrwd=\wd\z@\fi}
\def\findminarrwd#1#2#3{\@rrwd=#3\relax
   \measurelabel{#1}\measurelabel{#2}}
\def\findCDarrwd#1#2{\@rrwd=\minCDharrlen
   \measurelabel{#1}\measurelabel{#2}%
  }

\newcount\row \row=\@ne \newcount\col \col=\@ne 
 \newcount\numrows 
\numrows=\@ne
 \newcount\numcols
\newcount\arrspan \newdimen\vrtxhalfwd  \newbox\tempbox

\def\DANABUG{\advance\col by \@ne
 \@rrwd=\minCDharrlen
  \advance\@rrwd by \vrtxhalfwd
  \advance\@rrwd by \CDarrsurr
  \ifnum\col>\numcols \numcols=\col
     \newlocdimen{col\the\col}\locdimen=\@rrwd 
  \else \ifdim\@rrwd>\c@l \c@l=\@rrwd\fi\fi}

\def\drop#1\\{
  \findvrtxhalfsum\DANABUG\advance\row by 2 \measureinit}

\def\measureinit{\col=\@ne \vrtxhalfwd=-\CDarrsurr\arrspan=\@ne\@rrwd=\z@
   \setbox\tempbox=\hbox\bgroup$}
\def\measure{
  \let\harrow\measureCDarrow
  \let\CDCR=\measureCR 
   \findminc@lwd 
  \inmeasureCDtrue
  \row=\@ne \numcols=\z@ \measureinit}

\def\endmeasure{\findvrtxhalfsum\DANABUG
  \numrows=\row 
  \inmeasureCDfalse}

\def\newlocdimen#1{\advance\dimenc@unt by \@ne
  \ifnum\dimenc@unt<\insc@unt
     \else\errmessage{No room for the CD}\fi
  \dimendef\locdimen=\dimenc@unt
  \expandafter\dimendef\csname#1\endcsname=\dimenc@unt}

 \def\r@wc@l{\csname row\the\row col\the\col\endcsname}
 \def\c@l{\csname col\the\col\endcsname}

 \def\findvrtxhalfsum{$\egroup
  \newlocdimen{row\the\row col\the\col}
  \locdimen=\vrtxhalfwd 
  \vrtxhalfwd=0.5\wd\tempbox 
  \advance\vrtxhalfwd by \CDarrsurr
  \advance\locdimen by \vrtxhalfwd 
  \advance\@rrwd by \locdimen 
  \maybepush
  \divide\@rrwd by \arrspan\relax
  \ifdim\@rrwd<\minc@lwd
    \ifnum\col>\@ne \@rrwd=\minc@lwd\fi \fi
  \loop 
    \ifnum\col>\numcols \numcols=\col
       \newlocdimen{col\the\col}
       \locdimen=\@rrwd 
    \else \ifdim\@rrwd>\c@l \c@l=\@rrwd\fi \fi
   \ifnum\arrspan>\@ne
      \advance\arrspan by -1 \advance\col by \@ne
  \repeat }

 \def\measureCDarrow#1#2#3{\findvrtxhalfsum
   \arrspan=\sp@ncnt\relax\global\sp@ncnt=1\relax
   \advance\col by \@ne
   \findCDarrwd{#2}{#3}%
   \setbox\tempbox=\hbox\bgroup$}

 \newcount\dr@tn \dr@tn=\z@
 \def\locate#1:#2{\ifinmeasureCD\else
   \count@=-#1
   \multiply\count@ by 2
   \advance\count@ by #2
   \dimen@=\count@\@rrwd
   \ifnum\dr@tn=\@ne\relax \else\dimen@=-\dimen@ \fi
   \dimen@i=\@rrdp
   \ifnum\dr@tn>\z@\advance\dimen@i by \CDstrutlen \fi
   \dimen@i=\count@\dimen@i
   \count@=#2 \multiply\count@ by 2
   \divide\dimen@ by \count@
   \divide\dimen@i by \count@
   \lift\dimen@i\nudge\dimen@\fi}

\def\betweenCDrows{\advance\row by \@ne \col=\@ne
\options}

\def\hbegin{\hbox\bgroup\kern\c@l \kern-\r@wc@l$}
\def\hend{$\glet\maybepush\relax \CDstrut\egroup}
\def\vbegin{\setbox\tempbox=\hbox\bgroup$}
\def\vend{$\egroup\ht\tempbox=\z@\dp\tempbox\CDvarrlen
  \box\tempbox}
\def\setCD{\let\harrow=\setCDarrow
  \let\CDCR=\setCR 
  \row=\@ne \col=\@ne \hbegin}
\let\endsetCD=\hend 

\def\findarrwd{\@rrwd=\z@ \count@=\col \advance\count@ by\sp@ncnt
  \loop\ifnum\count@>\col \advance\count@ by -1
      \advance\@rrwd by\csname col\the\count@\endcsname\repeat}
\def\setCDarrow#1#2#3{\kern\CDarrsurr\advance\col by \@ne
  \findarrwd \advance\@rrwd by -\r@wc@l  
  \@rrdp=\z@ 
  \maybepush
  \advance\col by -\@ne \advance\col by \sp@ncnt \span@ne
  \hbox to \@rrwd{\options
   \@rrwd=\scalefactor\@rrwd\hss
   \hplace{#1}\ulabel{#2}\dlabel{#3}\hss}%
   \kern\CDarrsurr}

\newdimen\labspacei 
\newdimen\labspaceii 

\newdimen\@xisheight
  \@xisheight=\the\fontdimen22\textfont2
\newdimen\labelskip
  \labelskip=\the\fontdimen10\textfont3 
\newdimen\topofshaft
\newdimen\botofshaft
\newdimen\botofulabel
\newdimen\topofdlabel
\def\getlabeldims{
  \topofshaft=0.5\sh@ftdiam
  \botofshaft=\topofshaft
  \advance\topofshaft by \@xisheight  
  \advance\botofshaft by -\@xisheight  
  \botofulabel=\topofshaft
  \advance\botofulabel by \labelskip
  \topofdlabel=\botofshaft
  \advance\topofdlabel by \labelskip}

\def\ulabel{\ifnum\row=\@ne\let\next\ulabeli
   \else\let\next\ulabellap\fi\next}
\def\ulabeli#1{\vbox{
  \clap{\kern-\@rrwd$\scriptstyle#1$}%
  \kern\botofulabel}\maybeoffset}
\def\ulabellap#1{\vbox to \z@{\vss
  \clap{\kern-\@rrwd$\scriptstyle#1$}%
  \kern\botofulabel}\maybeoffset}
\def\dlabel{\ifnum\row=\numrows\let\next\dlabeli
   \else\let\next\dlabellap\fi\next}
\def\dlabeli#1{\vtop{\kern\topofdlabel
  \clap{\kern-\@rrwd$\scriptstyle#1$}%
  }\maybeoffset}
\def\dlabellap#1{\vbox to \z@{\kern\topofdlabel
  \clap{\kern-\@rrwd$\scriptstyle#1$}%
  \vss}\maybeoffset}
\def\rlabel#1{\vbox to \z@{\vss
  \rlap{\kern\labelskip$\scriptstyle#1$}%
  \vss\kern-\@rrdp}\maybeoffset}
\def\llabel#1{\vbox to \z@{\vss
  \llap{$\scriptstyle#1$\kern\labelskip}%
  \vss\kern-\@rrdp}\maybeoffset}
\def\swlabel#1{\vtop{\kern0.5\@rrdp
  \llap{$\scriptstyle#1$\kern\labelskip\kern-0.5\@rrwd}
  }\maybeoffset}
\def\nwlabel#1{\vbox{
  \llap{$\scriptstyle#1$\kern\labelskip\kern-0.5\@rrwd}%
  \kern-0.5\@rrdp}\maybeoffset}
\def\selabel#1{\vtop{\kern0.5\@rrdp
  \rlap{\kern0.5\@rrwd\kern\labelskip$\scriptstyle#1$}%
  }\maybeoffset}
\def\nelabel#1{\vbox{
  \rlap{\kern0.5\@rrwd\kern\labelskip$\scriptstyle#1$}%
  \kern-0.5\@rrdp}\maybeoffset}
\def\cplace#1{\vbox to \z@{\vss
  \clap{$#1$\kern-\@rrwd}%
  \kern-\@rrdp\vss}\maybeoffset}
\def\hplace#1{\hbox to \@rrwd{#1}\maybeoffset}
\def\vplace#1{\clap{\vbox to \z@{#1\kern-\@rrdp}}\maybeoffset}

\newdimen\nudgeamount \nudgeamount=\z@
\newdimen\liftamount \liftamount=\z@
\let\maybeoffset\relax
\newbox\offsetbox \newdimen\lastheight
\def\dooffset{
  \setbox\offsetbox=\lastbox \lastheight=\ht\offsetbox 
  \setbox\offsetbox=\vbox{\kern-\liftamount\box\offsetbox}%
  \ht\offsetbox=\lastheight
  \kern\nudgeamount\box\offsetbox\kern-\nudgeamount
  \global\nudgeamount=\z@ \global\liftamount=\z@
  \glet\maybeoffset=\relax}
\def\nudge#1{\ifinmeasureCD\else
  \global\advance\nudgeamount#1\relax
  \global\let\maybeoffset\dooffset\fi}
\def\lift#1{\ifinmeasureCD\else
  \global\advance\liftamount#1\relax
  \global\let\maybeoffset\dooffset\fi}

\def\findarrdp{\@rrdp=\CDvarrlen
  \ifnum\sp@ncnt@>1
    \advance\@rrdp by \CDstrutlen
    \multiply\@rrdp by \sp@ncnt@
    \advance\@rrdp by -\CDstrutlen \fi
 }

\def\varrow#1#2#3{\ifnum\sp@ncnt>\@ne 
     \sp@ncnt@=\sp@ncnt\relax\fi
  \findarrdp \@rrwd=\z@ 
  \kern\c@l
   \hbox to \z@{\options
   \@rrdp=\scalefactor\@rrdp
    \hss\vplace{#1}\llabel{#2}\rlabel{#3}\hss}%
  \global\advance\col by \@ne \setsp@n\@ne\@ne
  }

\def\novarrow{\varrow\vfill{}{}}

\def\tweenarrows#1{\findarrwd \findarrdp \setsp@n\@ne\@ne
  \rlap{\options\cplace{#1}}}

\def\usarrow #1#2#3{\dr@tn=\@ne
  \findarrwd \findarrdp \setsp@n\@ne\@ne 
  \rlap{\options\cplace{#1}\nwlabel{#2}\selabel{#3}}%
  \dr@tn=\z@}
\def\dsarrow #1#2#3{\dr@tn=\tw@
  \findarrwd \findarrdp \setsp@n\@ne\@ne 
  \rlap{\options\cplace{#1}\swlabel{#2}\nelabel{#3}}%
  \dr@tn=\z@}
 \def\@rrow#1{\csname #1@rrow\endcsname}
 \def\R@rrow{\harrow \rtarrfill}
 \def\L@rrow{\harrow \ltarrfill}
 \def\V@rrow{\varrow \dnarrfill}
 \def\A@rrow{\varrow \uparrfill}
 \def\SE@rrow{\dsarrow \SEarrow}
 \def\NW@rrow{\dsarrow \NWarrow}
 \def\SW@rrow{\usarrow \SWarrow}
 \def\NE@rrow{\usarrow \NEarrow}
 \def\DS@rrow{\dsarrow \dnslope}
 \def\US@rrow{\usarrow \upslope}
 \def\upslope{\find@xyargs
       \@linelen=\unitlength\@sline}
 \def\dnslope{\find@xyargs\@yarg=-\@yarg\relax
       \@linelen=\unitlength\@sline}

\newtoks\optionlist 
\optionlist={}
\let\options\relax
\def\dooptions{\the\optionlist\global\optionlist={}%
  \glet\options=\relax}
\def\option#1{\ifinmeasureCD\else
  \glet\options=\dooptions
  \global\optionlist=\expandafter{\the\optionlist\relax#1}\fi}
\def\wider#1{\ifinmeasureCD\else
   \option{\advance\@rrwd by #1}\fi}
\def\deeper#1{\ifinmeasureCD\else
   \option{\advance\@rrdp by #1}\fi}

{\def\\{\global\let\sptoken= }\\ }

\def\CR{\futurelet\nexttok\testCR}
\def\testCR{\ifx\nexttok\sptoken
   \let\next\eatspaceCR\else\let\next\CDCR\fi\next}
\def\eatspaceCR#1 {\CR}
\def\measureCR{\ifx\nexttok\endmeasure\let\nextCR\relax
    \else\let\nextCR\drop\fi\nextCR}
\def\setCR{\ifodd\row
  \ifx\nexttok\endsetCD\else\hend\betweenCDrows\vbegin\fi
  \else\vend\betweenCDrows\hbegin\fi}

\countdef\dimenc@unt=11
\def\CD#1\endCD{
   \begingroup\let\\=\CR
  \m@th\offinterlineskip
   \measure#1\endmeasure\null\,\vcenter{\setCD#1\endsetCD}\,
   \endgroup
    }

\ifx\@clnwd\undefined \nol@g\else\catcode`\ =14\relax\fi
 \font\@linefnt=line10 
 \newcount\@tempcnta
 \newcount\@tempcntb
 \newdimen\@tempdima
 \newdimen\@tempdimb
 \newdimen\@wholewidth
 \newdimen\@halfwidth
   \@wholewidth\fontdimen8\@linefnt \@halfwidth .5\@wholewidth
 \newdimen\unitlength
 \newcount\@xarg
 \newcount\@yarg
 \newcount\@yyarg
 \newbox\@linechar
 \newdimen\@linelen
 \newdimen\@clnwd
 \newdimen\@clnht
 \newif\if@negarg
 
 \def\@whilenoop#1{}

 \def\@whiledim#1\do #2{\ifdim #1\relax#2\@iwhiledim{#1\relax#2}\fi}

 \def\@iwhiledim#1{\ifdim #1\let\@nextwhile=\@iwhiledim 
         \else\let\@nextwhile=\@whilenoop\fi\@nextwhile{#1}}

 \def\@sline{\ifnum\@xarg< 0 \@negargtrue \@xarg -\@xarg \@yyarg -\@yarg
   \else \@negargfalse \@yyarg \@yarg \fi
 \ifnum \@yyarg >0 \@tempcnta\@yyarg \else \@tempcnta -\@yyarg \fi
 \ifnum\@tempcnta>6 \@badlinearg\@tempcnta0 \fi
 \ifnum\@xarg>6 \@badlinearg\@xarg 1 \fi
 \setbox\@linechar\hbox{\@linefnt\@getlinechar(\@xarg,\@yyarg)}%
 \ifnum \@yarg >0 \let\@upordown\raise \@clnht\z@
    \else\let\@upordown\lower \@clnht \ht\@linechar\fi
 \@clnwd=\wd\@linechar
 \if@negarg \hskip -\wd\@linechar \def\@tempa{\hskip -2\wd\@linechar}\else
      \let\@tempa\relax \fi
 \@whiledim \@clnwd <\@linelen \do
   {\@upordown\@clnht\copy\@linechar
    \@tempa
    \advance\@clnht \ht\@linechar
    \advance\@clnwd \wd\@linechar}%
 \advance\@clnht -\ht\@linechar
 \advance\@clnwd -\wd\@linechar
 \@tempdima\@linelen\advance\@tempdima -\@clnwd
 \@tempdimb\@tempdima\advance\@tempdimb -\wd\@linechar
 \if@negarg \hskip -\@tempdimb \else \hskip \@tempdimb \fi
 \multiply\@tempdima \@m
 \@tempcnta \@tempdima \@tempdima \wd\@linechar \divide\@tempcnta \@tempdima
 \@tempdima \ht\@linechar \multiply\@tempdima \@tempcnta
 \divide\@tempdima \@m
 \advance\@clnht \@tempdima
 \ifdim \@linelen <\wd\@linechar
    \hskip \wd\@linechar
   \else\@upordown\@clnht\copy\@linechar\fi}
 
 \def\@getlinechar(#1,#2){\@tempcnta#1\relax\multiply\@tempcnta 8
 \advance\@tempcnta -9 \ifnum #2>0 \advance\@tempcnta #2\relax\else
 \advance\@tempcnta -#2\relax\advance\@tempcnta 64 \fi
 \char\@tempcnta}
 
 \def\vector(#1,#2)#3{\@xarg #1\relax \@yarg #2\relax
 \@tempcnta \ifnum\@xarg<0 -\@xarg\else\@xarg\fi
 \ifnum\@tempcnta<5\relax
 \@linelen=#3\unitlength
 \ifnum\@xarg =0 \@vvector 
   \else \ifnum\@yarg =0 \@hvector \else \@svector\fi
 \fi
 \else\@badlinearg\fi}
 
 \def\@svector{\@sline
 \@tempcnta\@yarg \ifnum\@tempcnta <0 \@tempcnta=-\@tempcnta\fi
 \ifnum\@tempcnta <5
   \hskip -\wd\@linechar
   \@upordown\@clnht \hbox{\@linefnt  \if@negarg 
   \@getlarrow(\@xarg,\@yyarg) \else \@getrarrow(\@xarg,\@yyarg) \fi}%
 \else\@badlinearg\fi}
 
 \def\@getlarrow(#1,#2){\ifnum #2 =\z@ \@tempcnta='33\else
 \@tempcnta=#1\relax\multiply\@tempcnta \sixt@@n \advance\@tempcnta
 -9 \@tempcntb=#2\relax\multiply\@tempcntb \tw@
 \ifnum \@tempcntb >0 \advance\@tempcnta \@tempcntb\relax
 \else\advance\@tempcnta -\@tempcntb\advance\@tempcnta 64
 \fi\fi\char\@tempcnta}
 
 \def\@getrarrow(#1,#2){\@tempcntb=#2\relax
 \ifnum\@tempcntb < 0 \@tempcntb=-\@tempcntb\relax\fi
 \ifcase \@tempcntb\relax \@tempcnta='55 \or 
 \ifnum #1<3 \@tempcnta=#1\relax\multiply\@tempcnta
 24 \advance\@tempcnta -6 \else \ifnum #1=3 \@tempcnta=49
 \else\@tempcnta=58 \fi\fi\or 
 \ifnum #1<3 \@tempcnta=#1\relax\multiply\@tempcnta
 24 \advance\@tempcnta -3 \else \@tempcnta=51\fi\or 
 \@tempcnta=#1\relax\multiply\@tempcnta
 \sixt@@n \advance\@tempcnta -\tw@ \else
 \@tempcnta=#1\relax\multiply\@tempcnta
 \sixt@@n \advance\@tempcnta 7 \fi\ifnum #2<0 \advance\@tempcnta 64 \fi
 \char\@tempcnta}
\catcode`\ =10

\dol@g 
\catcode`\@=\active
\LaTeXfeatures

\newtheorem{theorem}{Theorem}[section]

\newtheorem{prop}[theorem]{Proposition}

\newtheorem{remark}[theorem]{Remark}
\newtheorem{definition}[theorem]{Definition}

\newtheorem{problem}[theorem]{Problem}
\newtheorem{defin}[theorem]{Definition}

\def\wt{\widetilde}

\def\al{Val}

\begin{document}

\large{

\title{Multi-sorted logic, models and logical geometry}

\maketitle

\begin{center}

\author{

 E.~Aladova$^{a,b}$,
      A.~Gvaramia$^{c}$,
 B.~Plotkin$^{d}$,
  T.~Plotkin$^{a}$}

\smallskip
 {\small
               $^{a}$ Bar Ilan University,

          5290002, Ramat Gan, Israel

               {\it E-mail address:} aladovael At mail.ru
            }

\smallskip
 {\small
               $^{b}$ Penza State University,

          440026, Krasnaya st. 40, Penza, Russia
}


  \smallskip
        {\small
               $^{c}$ Abkhazian State University,

           384904, Universitetskaya st. 1, Sukhumi, Abkhazia
        }

 \smallskip
        {\small

        $^{d}$ Hebrew University of Jerusalem,

          91904, Jerusalem, Israel

              {\it E-mail address:} plotkin At macs.biu.ac.il
        }

\end{center}

\begin{abstract}
Let $\Theta$ be a variety of algebras,  $(H, \Psi, f)$ be a model,
where $H$ is an algebra from $\Theta$, $\Psi$ is a set of relation
symbols $\varphi$,  $f$ is an interpretation of all  $\varphi$ in
$H$. Let $X^0$ be an infinite set of variables, $\Gamma$ be a
collection of all finite subsets in $X^0$ (collection of sorts),
$\widetilde\Phi$ be the multi-sorted algebra of formulas. These data
define a knowledge base $KB(H,\Psi, f)$. In the paper the notion
of isomorphism of knowledge bases is considered. We give
sufficient conditions which provide isomorphism of knowledge
bases. We also  study the problem of necessary
and sufficient conditions for isomorphism of two knowledge bases.
\end{abstract}

\section{Introduction}\label{S_Int}

Speaking about  knowledge we proceed from its representation in
three components.

(1) \emph{Description of  knowledge} presents a syntactical
component of  knowledge. From algebraic viewpoint description of
 knowledge is a set of formulas $T$ in the algebra of formulas
$\Phi(X)$, $X=\{x_1, \ldots , x_n\}$. Now we only note that
$\Phi(X)$ is one of domains of multi-sorted algebra
$\widetilde\Phi$ (detailed definition of $\widetilde\Phi$ see in
\cite{Seven}, \cite{Plotkin_AG}, \cite{PAP} and
Section~\ref{sec:mult}).

(2) \emph{Subject area of  knowledge} is presented by a model
$(H,\Psi,f)$, where $H$ is an algebra in fixed variety of algebras
$\Theta$, $\Psi$ is a set of relation symbols $\varphi$ and  $f$
is an interpretation of each $\varphi$ in $H$.

(3) \emph{Content of  knowledge} is a subset in $H^{n}$, where
$H^{n}$  is the Cartesian power of $H$. Each content of  knowledge
$A$ corresponds to the description of  knowledge
$T\subset\Phi(X)$, $|X|=n$. If we regard $H^{n}$ as an affine
space then  this correspondence can be treated geometrically (see
Section~\ref{sec:Galois}).

In order to describe the dynamic nature of  a knowledge base we introduce  two
categories: the category  of descriptions of
knowledge $F_{\Theta}(f)$ and the category  of knowledge contents
$LG_\Theta(f)$. These categories are defined using the machinery
of logical geometry (see Sections~\ref{sec:vc}, \ref{sec:svjaz},
or \cite{PP_LNCS}).

We shall emphasize that all of our notions are oriented towards an
arbitrary variety of algebras $\Theta$. Therefore, algebra, logic and
geometry of knowledge bases are related to this variety. Universal
algebraic geometry and logical geometry deal with  algebras $H$
from $\Theta$, while  logical geometry studies also arbitrary models $(H,
\Psi, f)$. Moreover, for each particular variety of algebras
$\Theta$ there are its own interesting problems and solutions.

The objective of the present paper is to study connections between
isomorphisms of knowledge bases and isotypeness of subject areas of
knowledge.

Varying  $\Theta$, we arrive to numerous specific problems. In
particular, if $\Theta$ is a variety of all quasigroups, it is
interesting to understand the connection between logical isotypeness  and
isotopy of quasigroups \cite{Gv}, \cite{Sm}.

The paper consists of two parts. In the first one the necessary
notions from logical geometry are introduced. In the second part,
logical geometry is considered in the context of knowledge bases.
In particular,  we describe conditions on the models which provide
an isomorphism of  corresponding knowledge bases.

\section{Basic notions}\label{sec:first}

\subsection{Points and affine spaces}

Let an algebra $H \in \Theta$ and a set $X = \{x_1, \ldots, x_n\}$
be given. A point $\overline a =(a_1, \ldots, a_n)$ can be
represented as the map $\mu: X \to H$ such that
$a_i=\mu(x_i)$. Denote by $H^n$ the affine space consisting of
such points.

Every map $\mu$ gives rise to the homomorphism  $\mu:
W(X) \to H$, where $W(X)$ is the free algebra over a set $X$ in the
variety $\Theta$. Thus, every affine space can be considered as
the set $Hom(W(X),H)$ of all homomorphisms from $W(X)$ to  $H$.

Each point $\mu$  as a homomorphism has a kernel $Ker(\mu)$, which
is a binary relation on the set $W(X)$. By definition, elements
$w,w'\in W(X)$ belong to the binary relation $Ker(\mu)$ if and
only if $w^\mu = {(w')}^\mu$, where $w^\mu$ is notation for
$\mu(w)$.

We will also consider a logical kernel $LKer(\mu)$ of a point
$\mu$. A formula $u \in \Phi(X)$ belongs to $LKer(\mu)$, if the
point $\mu$ satisfies the formula $u$.

\subsection{Extended boolean algebras}\label{sec:eb}

We start from the definition of  an existential quantifier on a
boolean algebra. Let $B$ be a boolean algebra. Existential
quantifier on $B$ is a unary operation $\exists : B \to B$ such
that the following conditions hold:
\begin{enumerate}
\item $\exists \ 0 = 0$,
\item $a \leq \exists a$,
\item $\exists (a \wedge \exists b) = \exists a \wedge \exists b$.
\end{enumerate}

Universal quantifier $\forall : B \to B$ is dual to $\exists : B
\to B$, they are related by  $ \forall a=\neg(\exists (\neg a))$.

\begin{definition}\label{Def_ExtendedBA}
Let a set of variables $X=\{x_1,\dots ,x_n\}$ and a set of
relations $\Psi$ be given. A boolean algebra $B$ is called an
extended boolean algebra over $W(X)$ if

1. the existantial quantifier $\exists x$ is defined on $B$ for
all  $x \in X$, and  $\exists x \exists y = \exists y \exists x$
for all $x,y \in X$;

2. to every  relation symbol $\varphi \in \Psi$ of arity
$n_\varphi$ and a collection of elements $w_1,\ldots,
w_{n_\varphi}$ from $W(X)$ there corresponds a nullary operation
(a constant) of the form $\varphi(w_1,\ldots, w_{n_\varphi})$ in
$B$.
\end{definition}

Thus, the signature $L_X$ of extended boolean algebra consists
of the boolean connectives, existential quantifiers  $\exists x$
and of the set of constants $\varphi(w_1,\ldots, w_{n_\varphi})$:
$$L_X = \{\vee, \wedge, \neg, \exists x, M_X \},$$
where $M_X$ is the set of all $\varphi(w_1, \ldots, w_{n_\varphi})$.

The algebra of formulas $\Phi(X)$ is the example of an extended
boolean algebra (see \cite{Seven}, \cite{Plotkin_AG}, \cite{PAP}).
A formula $w\equiv w'$ is one of the constants, where $\varphi$ is
the equality predicate "$\equiv$". Depending on the context, we
call it equality or equation.

Consider another important example of  extended boolean
algebras. Let  $(f)=(H, \Psi,f)$ be a model. Take the affine space
$Hom(W(X),H)$ and denote by $Bool(W(X),H)$ the boolean algebra of
all subsets of $Hom(W(X),H)$.

Let us define on this algebra the existential quantifier. If $A$
is an element of  $Bool(W(X),H)$ then the element $B=\exists x A$
is defined by the rule: a point $\mu$ belongs to $B$ if there
exists a point $\nu \in A$ such that $\mu(x') = \nu(x')$ for each
$x' \in X$, $x' \neq x$.

Define now constants on $Bool(W(X),H)$. For a relational symbol
$\varphi$ of arity $m$ denote by $[\varphi(w_1, \ldots,
w_m)]_{(f)}$ the subset in $Bool(W(X),H)$ consisting  of all
points $\mu:W(X) \to H$ satisfying the relation $\varphi(w_1, \ldots,
w_m)$. This means that $(w_1^\mu, \ldots, w_m^\mu)$ belongs to the
set $f(\varphi)$, where $f(\varphi)$ is a subset in $H^m$,
consisting of all points which  belong to $\varphi$ under
interpretation $f$.

Denote this extended boolean algebra by $Hal_\Theta^X(f)$. In
particular, if $\Psi$ consists solely of the equality predicate symbol,
then the corresponding algebra is denoted by $Hal_\Theta^X(H)$.

In Section~\ref{sec:mult} we will define a homomorphism between
$\Phi(X)$ and  $Hal_\Theta^X(f)$:
$$
Val_{(f)}^X : \Phi(X) \to Hal_\Theta^X(f),
$$with the property
$$
Val_{(f)}^X(\varphi(w_1, \ldots, w_m))=
[\varphi(w_1, \ldots, w_m)]_{(f)}.
$$
This homomorphism allows us to define algebraically such notions as
''a point satisfies a formula'' and ''a logical kernel of a
point''. Such approach agrees with the model theoretic
inductive one (see \cite{Marker}).

Now we only observe, that for a formula $u \in \Phi(X)$ its image
$Val_{(f)}^X (u)$ is defined as the set of points  $\mu:W(X) \to H$
satisfying $u$.
In this case, a formula $u\in \Phi(X)$ belongs to the logical
kernel $LKer(\mu)$ of $\mu:W(X) \to H$ if and only if $\mu\in
Val_{(f)}^X(u)$. Note also, that $LKer(\mu)$ is a boolean
ultrafilter in the algebra of formulas $\Phi(X)$ containing
$X$-elementary theory  $Th^X(f)$ of the model $(H, \Psi, f)$. In
this sense, we say that $LKer(\mu)$ is an $LG$-type of the point
$\mu$ (see \cite{Marker} for the model theoretic definition of a
type and  \cite{Zhitom_types} for $LG$-type). Recall that
$Th^X(f)$ consists of all formulas $u \in \Phi(X)$ which hold true on each point
$\mu: W(X) \to H$. Thus,
$$
Th^X(f)=\bigcap_{\mu \in Hom(W(X),H)}LKer(\mu).
$$

\subsection{Galois correspondence}\label{sec:Galois}

 Define now a correspondence between sets $T$ of formulas of the form $w\equiv
w'$ in the algebra of formulas $\Phi(X)$ and subsets of points $A$
from the affine space $Hom(W(X),H)$. We set
 $T_{H}' = A$, where $A$ consists of all points $\mu: W(X) \to H$ such that $T
\subset Ker(\mu)$. In other words, $T_{H}'$ consists of all points
satisfying all formulas from $T$. We call this $T_{H}'$ \emph{an
algebraic set} defined by the set of formulas $T$.

On the other hand, for a given set of points $A$ we define a set
of formulas  $T$ as
$$
 T=A_{H}'=\bigcap_{\mu \in A} Ker(\mu).
$$
By the definition, $T$ is a congruence, it is called
\emph{$H$-closed congruence} defined by $A$. One can check that
such correspondence between  sets of formulas of the form
$w\equiv w'$ from the algebra $\Phi(X)$ and sets of points from
the affine space  $Hom(W(X),H)$ is the Galois correspondence
\cite{MacLane}.

Now we consider the case of arbitrary set of formulas
$T\subset\Phi(X)$. Let $T_{(f)}^L = A$ be a set of all points
$\mu: W(X) \to H$ such that $T \subset LKer(\mu)$. The set
 $T_{(f)}^L$ is called \emph{a definable set}
presented by the set of formulas $T$. Let now $A$ be a set of
points from $Hom(W(X),H)$. We define $A_{(f)}^L$ as
$$
 A_{(f)}^L=T=\bigcap_{\mu \in A}
LKer(\mu).
$$
Direct calculations show that  $u\in A_{(f)}^L$ if and only if
$A\subset Val_{(f)}^X(u) $. Note that  $A_{(f)}^L$ is a filter in
$\Phi(X)$ called \emph{$H$-closed filter} defined by the
set $A$.

Thus, the  Galois correspondences described above give rise to
universal algebraic geometry if $T$ is a set of equalities, and
to logical geometry if $T$ is an arbitrary set of formulas.

Recall that  a set $A$ from $Hom(W(X),H)$ is Galois-closed if
$A^{''}_{H}=A$ or $A^{LL}_{(f)}=A$, depending on the given Galois
correspondence. A congruence $T$ on $W(X)$ is Galois-closed if
 $T^{''}_{H}=T$, a filter $T$ in $\Phi(X)$ is Galois-closed if $T^{LL}_{(f)}=T$.

 So, we have a one-to-one correspondence between algebraic sets in $Hom(W(X),H)$ and closed congruences on
 $W(X)$, between definable sets in $Hom(W(X),H)$  and closed filters in the extended boolean algebra $\Phi(X)$.

\subsection{Some categories}\label{sec:vc}

In this section we define various categories, which are necessary
for  further considerations.

\subsubsection{Categories $\Theta^0$,  $\widetilde \Phi$ and $\Theta^{*}(H)$.}
Let an infinite set of variables $X^0$ and a collection $\Gamma$
of finite subsets of $X^0$ be given.

Denote by $\Theta^0$ the category of all free algebras $W(X)$ in
$\Theta$, $X \in \Gamma$. Morphisms in this category are
homomorphisms $s : W(X) \to W(Y)$.

Along with free algebras $W(X)$ we consider algebras of formulas
$\Phi(X)$, which are also associated  with the variety $\Theta$. We
define a category $\widetilde \Phi$ of all $\Phi(X)$, $X \in
\Gamma$ in such a way that to each morphism $s : W(X) \to W(Y)$ it
corresponds a morphism $s_\ast : \Phi(X) \to \Phi(Y)$ and this
correspondence gives rise to a covariant functor from $\Theta^0$
to $\widetilde \Phi$.

Define now the category $\Theta^{*}(H)$ of affine spaces over
$H\in \Theta$. Objects of this category are affine spaces
$Hom(W(X),H)$, morphisms are maps:
$$
\widetilde s: Hom(W(X),H) \to Hom(W(Y),H),
$$
where
$$
s: W(Y)\to W(X)
$$
are morphisms in the category of free algebras $\Theta^{0}$.

For a point $\mu: W(X)\to H$ the point $\nu=\widetilde s(\mu):
W(Y)\to H$ is defined as follows:
$$
\widetilde s(\mu)=\mu s : W(Y)\to H,
$$
that is,  $\nu(w)=\mu(s(w))$, $w\in W(Y)$.

Passages $W(X)\to Hom(W(X),H)$ and  $s\to \widetilde s$ give rise
to a contravariant functor
$$
\Theta^0 \to \Theta^{*}(H).
$$

There is the following
\begin{theorem}[\cite{MPP1}]
The functor $ \Theta \to \Theta^{*}(H) $ defines a duality of
categories if and only if the algebra $H$ generates the variety of
algebras $\Theta$, i.e., $\Theta=Var(H)$.

\end{theorem}

\subsubsection{Categories $Hal_{\Theta}(H)$ and $Hal_{\Theta}(f)$.}\label{sub:CategHal}

For a given model $(H,\Psi,f)$ we define categories
$Hal_{\Theta}(H)$ and $Hal_{\Theta}(f)$. The first category is
related to universal algebraic geometry, while the second one to
logical geometry.

Objects of these categories are algebras $Hal_{\Theta}^{X}(H)$ and
$Hal_{\Theta}^{X}(f)$, respectively. The categories
$Hal_{\Theta}(H)$ and $Hal_{\Theta}(f)$ have different objects,
since the sets of constants in algebras $Hal_{\Theta}^{X}(H)$ and
$Hal_{\Theta}^{X}(f)$ are different (see Section~\ref{sec:eb}).

Denote morphisms for both categories by $s_\ast$. A homomorphism
$s:W(Y)\to W(X)$ gives rise to a map
$$
\widetilde s: Hom(W(X),H) \to Hom(W(Y),H).
$$
In its turn, $ \widetilde s$ defines a morphism
$$ s_*:
Bool(W(Y),H)\to Bool(W(X),H)
$$
by the rule: for an arbitrary $B\subset Hom(W(Y),H)$ we put
$$
s_{*}B=\widetilde s ^ {-1} (B)=A\subset Hom(W(X),H).
$$
Thus, $A$ is a  full pre-image of $B$ under $\widetilde s$, it
consists of all points $\mu$ from $Hom(W(X),H)$ such that
$\widetilde s (\mu)=\mu s  \in B$.

We would like to link together categories $\widetilde \Phi$ and
$Hal_{\Theta}(f)$. Let $s:W(Y)\to W(X)$,
$s_\ast:\Phi(Y)\to\Phi(X)$ and $v\in\Phi(Y)$ be given.

\begin{prop}\label{pr:romb}
A point $\mu:W(X)\to H$ satisfies  the formula  $u=s_\ast v$ if and
only if  $\mu s$ satisfies the formula $v$.
\end{prop}

\begin{proof}
In fact, this result follows from axiom (5) in
Definition~\ref{ha:ms}, which regulates the action  of morphism
$s_\ast$ on formulas of the form $\varphi(w_1,\dots, w_m)$.
These formulas  generate freely the algebra  $\widetilde \Phi$ as a
multi-sorted  algebra (see Section~\ref{sec:mult} or
\cite{Seven}).
\end{proof}

Let $A$ be the set of all points satisfying the formula $u=s_\ast v\in
\Phi(X)$, $B$ be the set of all points satisfying the formula
$v\in\Phi(Y)$.

\begin{prop}\label{pr:pro}
Let $A_0=s_\ast B=\widetilde s^{-1}(B)$. Then $A_0=A$.
\end{prop}
\begin{proof}
Let $\mu'\in Hom(W(X),H)$ belongs to $\widetilde s ^ {-1}
(B)=A_0$. By the definition this means that $\widetilde
s(\mu')=\mu' s\in B.$  Thus, $\mu'  s$  satisfies the formula $v$.
By Proposition~\ref{pr:romb}, the point $\mu'$ satisfies the
formula $u$. Hence, $\mu'\in A$.
 \end{proof}

We call a set $A$  \emph{$s$-closed} if $A_0=A$, that is,
$s_\ast(\widetilde s A)=A$. As follows from
Proposition~\ref{pr:pro}, each  definable set is $s$-closed.

Consequently, we have the commutative diagram

 \begin{equation}\label{diag:1}
\CD
\Phi(Y) @> s_\ast >> \Phi(X)\\
@V \al_{(f)}^Y  VV @VV \al_{(f)}^X V\\
Hal_\Theta (f_2)@> s_\ast=\widetilde{s}^{-1} >>Hal_\Theta (f_1),
\endCD
\end{equation}

The commutativity of this diagram means that if $v\in \Phi(Y)$,
$u=s_\ast v \in \Phi(X)$, $A=\al_{(f)}^X(u)$, $B=\al_{(f)}^Y(v)$,
then  $\al_{(f)}^X(s_\ast v)=s_\ast\al_{(f)}^Y(v)$.

From the categorical viewpoint, commutative diagram~(\ref{diag:1})
determines a covariant functor from $\widetilde \Phi$ to
$Hal_{\Theta}(f)$. From the point of view of multi-sorted
algebras, the last equality means that $Val_{(f)}$ is a
homomorphism of multi-sorted Halmos algebras.

\subsubsection{Categories $AG_{\Theta}(H)$ and $LG_{\Theta}(f)$.}
The first category is related to algebraic sets in universal
algebraic geometry, while the second one to definable sets in logical
geometry.

Objects of the category $AG_\Theta(H)$ are partially ordered sets $AG_\Theta^X(H)$ of all
algebraic sets in $Hom(W(X),H)$ with fixed $X$. 

For a given homomorphism $s:W(Y)\to W(X)$, a morphism
$$
\widetilde s_\ast: AG_\Theta^X(H)\to AG_\Theta^Y(H)
$$
is defined as follows. Let $A$ be an algebraic set in
$Hom(W(X),H)$. Then $\widetilde s_\ast A=B$ is an algebraic set
determined by the set of points of the form $\nu=\mu s$, where
$\mu\in A$. In other words, $B$ is the Galois closure of this set
of points, i.e.,  $B= \widetilde s_\ast A= (\widetilde s
A)''_{H}$.  Morphisms, defined in such a way, preserve the partial
order relation.

Objects of $LG_\Theta(f)$ are sets of all definable sets in
$Hom(W(X),H)$ with fixed $X$. We assume, that each object
$LG_\Theta^X(f)$ is a lattice.

Define morphisms in $LG_\Theta(f)$ as:
$$
\widetilde s_\ast: LG_\Theta^X(H)\to LG_\Theta^Y(H).
$$
Let $A$ be a definable  set in $Hom(W(X),H)$. Then $\widetilde
s_\ast A=B$ is a definable set given by the set of points of the
form $\nu=\mu s$, where $\mu\in A$. In other words, $B$ is the
Galois closure of this set of points, i.e., $B= \widetilde s_\ast
A= (\widetilde s A)^{LL}_{(f)}$.

\subsubsection{Categories $C_\Theta (H)$ and $F_\Theta(f)$.}

Objects of  $C_\Theta (H)$ are partially ordered sets of
$H$-closed congruences on $W(X)$. They are in one-to-one
correspondence with the objects $AG_\Theta^X(H)$. Morphism in
$C_\Theta (H)$
 $$
 \widehat s_\ast: C_\Theta^Y(H) \to C_\Theta^X(H)
 $$
is defined using the maps between $H$-closed congruences in
$C_\Theta^Y(H)$ and $C_\Theta^X(H)$. Let $T_2$ be an $H$-closed
congruence  in $C_\Theta^Y(H)$. Specify $T_1$ as an $H$-closed
congruence in $C_\Theta^X(H)$ defined by the set of all
equations of the form $s_\ast(w\equiv w')$, for all  $w\equiv w'$
from $T_2$. In other words, $T_1=(s_\ast T_2)''_{H}$.

Objects of the category  $F_\Theta(f)$ are lattices of $H$-closed
filters. We define morphisms in $F_\Theta (H)$
 $$
 \widehat s_\ast: F_\Theta^Y(H) \to F_\Theta^X(H),
 $$
using the maps between  $H$-closed filters in $F_\Theta^Y(H)$ and
$F_\Theta^X(H)$. Let $T_2$ be an $H$-closed filter in
$F_\Theta^Y(H)$. Determine $T_1$ as  the $H$-closed filter in
$F_\Theta^Y(H)$ defined by the set of formulas of the form
$s_\ast v$, for all  $v$ from $T_2$, that is, $T_1=(s_\ast
T_2)^{LL}_{(f)}$.

\subsection{Relation between categories $LG_\Theta(f)$ and
$F_\Theta(f)$}\label{sec:svjaz}

We would like to determine the duality of categories
$LG_\Theta(f)$ and $F_\Theta(f)$. According to their Galois
correspondence there is a one-to-one correspondence between objects
of these categories.

Let a homomorphism  $ s:W(Y)\to W(X)$ and a definable set $B_0$ from
$Hom (W(Y),H)$ be given.

Define the  set $A_0$ as the full pre-image of $B_0$ under $\widetilde s$, i.e., $A_0=\widetilde s^{-1} (B_0)$ (see
Section~\ref{sub:CategHal}). Let $B$ be a definable set such that
$B= \widetilde s_\ast A_0= (\widetilde s A_0)^{LL}$. Since
$\widetilde s A_0\subset B$, then  $B=(\widetilde s
A_0)^{LL}\subset B_0^{LL}=B_0$.

Define the $H$-closed filter $T_2$ as $T_2=B^L$. Then,
$s_\ast$ and  $T_2$ determine the $H$-closed filter $T_1=(s_\ast
T_2)^{LL}=\widehat s_\ast T_1.$ Finally, we put $A=T_1^L$.

There is the  commutative diagram: 

 \begin{equation}\label{d:2}
\CD
T_2 @> \widehat s_\ast >> T_1\\
@V \al_{(f)}^Y  VV @VV \al_{(f)}^X V\\
B @< \widetilde{s}_\ast << A
\endCD
\end{equation}

Indeed, since objects $A_0,B, T_2,T_1$ are defined uniquely by $B_0$ and
$s:W(Y)\to W(X)$, for the commutativity of the
diagram it is enough to check that $A_0=A$. But this equality
follows from Proposition~\ref{pr:pro}.

Moreover, $\mu\in A_0$ if and only if $\mu s\in B$. In its turn, $\mu s\in B$
 if and only if $\mu s $
satisfies each formula  $v\in T_2$. By Proposition~\ref{pr:romb},
$\mu s $ satisfies $v\in T_2$ if and only if  $\mu$ satisfies
$u=s_\ast v\in T_1$. Since $A$ consists of all points satisfying
all  formulas $u=s_\ast v\in T_1$ then $A_0=A$.

From  diagram~(\ref{d:2}) follows that for each formula  $v\in
T_2$ there is the relation
$$
Val^Y_{(f)}=\widetilde s_\ast Val^X_{(f)}s_\ast.
$$

 \begin{defin}\label{def:gr}
A map $\alpha: A\to B$ of definable sets is called {\it
generalized regular} if there is a map $\widetilde{s}_\ast:A\to
B$ satisfying commutative diagram~(\ref{d:2}) such that
$\alpha(\mu)=\wt s_\ast(\mu)$, for all $\mu\in A$.
  \end{defin}

By the definition of the map $\widetilde{s}_\ast$, the image of a
definable set under generalized regular map is a definable set.
Thus, $LG_\Theta(f)$ is the category of lattices of definable sets
with generalized regular maps as morphisms.

The similar approach works for the category of algebraic sets
$AG_\Theta(H)$. So, we have a particular case of
diagram~(\ref{d:2}):

 $$
\CD
T_2 @> s_\ast >> T_1\\
@V \al_{H}^Y  VV @VV \al_{H}^X V\\
B @< \widetilde{s}_\ast << A,
\endCD
$$
where $T_1$ and $T_2$ are the Galois-closed congruences.

\begin{defin}
A map $\alpha: A\to B$ of algebraic sets is called {\it regular}
if there is a map $\wt s_\ast:A\to B$ satisfying the commutative
diagram above such that $\alpha(\mu)=\wt s_\ast(\mu)$, for all
$\mu\in A$.
\end{defin}

Thus, $AG_\Theta(H)$ is the category of partially ordered algebraic
sets with regular maps as morphisms.

Summarizing, we have the theorem.

\begin{theorem}\label{th:anti}
Let  $Var(H)=\Theta$. The category $F_\Theta(f)$ of lattices of
$H$-closed filters is anti-isomorphic to the category $LG_\Theta(f)$
of lattices of definable sets. The category $C_\Theta (H)$ of
partially ordered congruences is anti-isomorphic to the category
$AG_\Theta(H)$ of partially ordered algebraic sets.
\end{theorem}

\begin{proof}
The proof of Theorem \ref{th:anti} follows from diagram~(\ref{d:2}).
The condition $Var(H)=\Theta$ ensures that the ho\-mo\-morphism $ s:W(Y)\to W(X)$ uniquely  defines 
morphism $\widetilde
s_\ast$.

\end{proof}

\subsection{Multi-sorted Halmos algebras}\label{sec:mult}

In Section~\ref{sub:CategHal} we defined the categories
$Hal_\Theta(H)$ and $Hal_\Theta(f)$. There is a natural way to treat
these categories as multi-sorted algebras (see \cite{P9},
\cite{PAP}). We put
$$
 Hal_\Theta(H)=(Hal_{\Theta}^{X}(H), X\in
\Gamma),
$$
$$
Hal_\Theta(f)=(Hal_{\Theta}^{X}(f), X\in \Gamma).
$$
In this case, objects of the categories are presented as domains
of multi-sorted algebras, while morphisms $s_{*}$ are unary
operations between domains. These  algebras are Halmos algebras.

\begin{remark}\label{rm:def}
We widely use the name P. Halmos, because he was one of the
creators of algebraic logic. He introduced the important notion of
a polyadic algebra. Along with other notions of universal algebra
and universal algebraic geometry, the notion of a polyadic algebra
gave rise to the theory, which, in particular, is used in this
paper.
\end{remark}

For the precise definition of a multi-sorted Halmos algebra, first of
all,  we specify a signature of such algebras.

Let a finite set $X$ from $\Gamma$, a variety $\Theta$,  an algebra
$H\in \Theta$ and  a set of relation symbols  $\Psi$ be given.
The signature $L^\Psi$ of a multi-sorted Halmos algebra
$\mathfrak{L}=(\mathfrak{L}_X, X \in \Gamma)$ includes the
signature of extended boolean algebras  $L_X$ (see
Section~\ref{sec:eb}) and operations of the form $s_\ast:
\mathfrak{L}_X \to \mathfrak{L}_Y$, which correspond to morphisms
$s:W(Y) \to W(X)$ in $\Theta^0$.

\begin{definition}\label{ha:ms}
A multi-sorted algebra $\mathfrak{L}=(\mathfrak{L}_X, X \in \Gamma)$
in the signature $L^\Psi$ is a Halmos algebra if
\begin{enumerate}
\item
 Each domain $\mathfrak{L}_X$ is an extended boolean algebra in the signature $L_X$.

\item %
Each map  $s_*: \mathfrak{L}_X \to \mathfrak{L}_Y$ is a
homomorphism of boolean algebras.

\item %
For given $s_{1*}: \mathfrak{L}_X\to \mathfrak{L}_Y$ and  $s_{2*}:
\mathfrak{L}_Y\to \mathfrak{L}_Z$ there is the equality:
 $$
 s_{1*}s_{2*}=(s_1s_2)_\ast.
 $$
In other words, it means that the correspondence $W(X)\to
\mathfrak{L}_X$ and  $s\to s_\ast$ define a covariant functor
from the category  $\Theta^0$ to the category $\mathfrak{L}$.

\item %
Next two axioms control the interaction of $s_{*}$ with
quantifiers:

 \begin{itemize}
 \medskip
\item[(a)] %
Let  $s_1: W(X)\to W(Z)$ and   $s_2: W(X)\to W(Z)$ be given.
Suppose, that $s_1( y) = s_2(y)$ for all $y\neq x$, $x$, $y\in X$.
Then
$$
s_{1*} \exists x a = s_ {2*} \exists x a, \ a \in
\mathfrak{L}_X.
$$

 \medskip

 \item[(b)]
Let $s: W(X)\to W(Y)$ and $s(x)=y$ be given, $x\in X$, $y\in Y$.
Let $x'\neq x$, $x'\in X$. Suppose, that $s(x')=w$, where $w\in
W(Y)$, and $y$ does not belong to the support of $w$. This
condition means, that $y$ does not participate in the shortest
expression of the element $s(x')\in W(Y)$. Then
$$
 s_{*}(\exists x
a) = \exists (s(x)) (s_*a),\ a \in \mathfrak{L}_X.
$$
\end{itemize}

 \item
Let a relation symbol $\varphi \in \Psi$ of arity $m$ and $s:
W(X)\to W(Y)$  be given. Then
$$
s_\ast (\varphi(w_1,\ldots, w_m))=\varphi(sw_1,\ldots,sw_m).
$$

In particular,  for each equation $w\equiv w'$ we have
 $$
 s_{*}(w\equiv w')=(s(w)\equiv s(w')).
 $$

\end{enumerate}
\end{definition}

Halmos algebras constitute  a variety, denote it by $Hal_\Theta$.
Moreover, the following fact takes place.
\begin{theorem}\label{th:var}
Let a model  $(f)=(H,\Psi,f)$, $H\in \Theta$, be given. The
variety $Hal_\Theta$ is generated by all algebras $Hal_\Theta(f)$
for all $H \in \Theta$.
\end{theorem}

Now we give a more precise definition of the algebra $\widetilde
\Phi=(\Phi(X), X\in \Gamma)$ and homomorphism $Val_{(f)}$.
For the detailed constructions of $\widetilde \Phi$ and $Val_{(f)}$ see
\cite{Seven}, \cite{Plotkin_AG}, \cite{Pl_GAGTA}, \cite{PAP}.

Let  $\varphi$ denote a relation symbol of arity $m$ from $\Psi$,
$M_X$ be the set of all $\varphi(w_1, \ldots,w_m)$, $w_i\in W(X)$.

The algebra $\widetilde \Phi= (\Phi(X), X \in \Gamma)$ is the free
algebra generated by multi-sorted set  $M=(M_X, X \in \Gamma)$ in
the variety $Hal_\Theta$.

For each $X$ we define a map
$$
 M_X \to Hal_\Theta^X(f)
$$
by the rule
$$
\varphi(w_1, \ldots,w_m) \to [\varphi(w_1,
\ldots,w_m)]_{(f)}.
$$
It induces the map of multi-sorted sets
$$
M \to Hal_\Theta(f).
$$
Since $M$  generates  freely the algebra $\widetilde \Phi$, then the
last map can be extended up to the homomorphism of multi-sorted
algebras
$$
 Val_{(f)}: \widetilde \Phi \to Hal_\Theta(f).
 $$
On components we have
$$
Val_{(f)}^X: \Phi(X) \to Hal_\Theta^X(f).
$$

Note that the algebra $\widetilde \Phi$ can be defined
semantically. Let $\mathfrak L^{0}$ be an absolutely free algebra
generated by the set  $M$ in the signature $L^\Psi$. Let a model
$(f)=(H,\Psi, f)$ and the corresponding algebra $Hal_\Theta(f)$ be
given. We will treat the algebras $Hal_\Theta(f)$ as universal
realization of the algebra $\mathfrak L^{0}$.

With each element $\varphi(w_1, \ldots,w_m)\in\mathfrak L^{0} $ we
associate an element $[\varphi(w_1, \ldots,w_m)]_{(f)}\in
Hal_\Theta(f)$. This correspondence gives rise to the ho\-mo\-morphism
$$
Val^0_{(f)}: \mathfrak L^{0} \to Hal_\Theta(f).
$$

Denote by $\rho_{(f)}$ the kernel of this homomorphism. Note that
it coincides with the set of identities of the algebra
$Hal_\Theta(f)$. Let us consider the congruence
$$
 \rho=\bigcap_{(f)}
\rho_{(f)}.
$$
Since $Val^0_{(f)}$ is a unique homomorphism from $\mathfrak
L^{0}$ to $Hal_\Theta(f)$ and all algebras of the form
$Hal_\Theta(f)$ generate the variety $Hal_\Theta$, then
$$
\widetilde \Phi= \mathfrak L^{0}/\rho.
$$
This expression gives rise to a description of the algebra
$\widetilde \Phi$ which allows us to calculate  images of the
elements from $\mathfrak L^{0}$ in the algebra $Hal_\Theta(f)$. In
this sense, this is a semantical definition of  $\widetilde \Phi$.

All above can be summarized in the diagram
$$
\CD
\mathfrak L ^{0} @[2]> Val^0_{(f)}  >> Hal_\Theta(f)\\
 @[2]/SE/\rho //@.@.\;    @/NE// Val_{(f)}/\\
 @.\widetilde{\Phi}
\\
\endCD
$$

\section{Logical geometry and knowledge bases}

\subsection{From logic and geometry to knowledge theory}

In the previous section we introduced a necessary system of
notions. All these concepts naturally arise and interact in a
certain order. The further exposition will be related to
applications to knowledge bases.

\subsection{Knowledge bases}

From now on we will treat  categories  $F_\Theta(f)$ and
$LG_\Theta(f)$ as {\it the categories of description of a knowledge
and content of a knowledge}, accordingly.

Recall that we distinguished three components of  knowledge
rep\-re\-sentation:
\begin{itemize}
\item \emph{description of  knowledge,}

\item \emph{subject  area of  knowledge,}

\item \emph{content of knowledge.}
\end{itemize}

The next three mathematical objects correspond to these
components:

\begin{itemize}
\item \emph{the category of lattices of $H$-closed filters}
$F_\Theta(f)$,

\item \emph{a model }  $(H,\Psi,f)$,

\item \emph{the category of lattices of definable sets}
$LG_\Theta(f)$.
\end{itemize}

\begin{defin}\label{df:kno}
A knowledge base $KB = KB(H,\Psi,f)$ is a triple
$(F_\Theta(f),LG_\Theta(f),Ct_f)$, where $F_\Theta(f)$ is the
category of description of  knowledge, $LG_\Theta(f)$ is the
category of content of knowledge, and
$$
Ct_f: F_\Theta(f) \to
LG_\Theta(f)
$$
is a contravariant functor.
\end{defin}
The functor $Ct_f$ transforms the knowledge description to the
knowledge content. Morphisms of the categories
$F_\Theta(f)$ and $LG_\Theta(f)$ make  knowledge bases a
dynamical object.

\begin{remark}
We use the term ''knowledge bases'' instead of a more precise ''a
knowledge base model''.
\end{remark}

For a given model $(f)=(H,\Psi,f)$, each concrete knowledge is a
triple $(X,T, A)$, where  $X\in \Gamma$,  $T$ is a set of
formulas from $\Phi(X)$ and  $A$ is the set of points from
$Hom(W(X),H)$ such that $A=T^L_{(f)}=(T_{(f)}^{LL})^L$. Therefore,
$T$ and $T_{(f)}^{LL}$ describe the same content $A$.

\subsection{Isomorphism of knowledge bases}

The definition of an iso\-morphism of two knowledge bases $KB_1$ and
$KB_2$ assumes an isomorphism of categories of knowledge content,
which implies the isomorphism of ca\-te\-gories of descriptions of
knowledge $F_\Theta(f_1)$ and  $F_\Theta(f_2)$. Thus,

\begin{defin}\label{d:inaut}
Knowledge bases $KB_1 = KB(H_1,\Psi,f_1)$ and $KB_2=
KB(H_2,\Psi,f_2)$ are called isomorphic if they match the
commutative diagram
$$ \CD
 F_\Theta(f_1)@>\alpha>> F_\Theta(f_2)\\
@V Ct_{f_1} VV @VV Ct_{f_2}V\\
LG_\Theta(f_1)@>\beta>> LG_\Theta(f_2),\\
\endCD
$$
where $\alpha$ and  $\beta$  are isomorphisms of categories.
\end{defin}

Let us return to the ideas of logical geometry with respect to
knowledge bases. We will use some material from \cite{APP}.

\begin{defin}\label{def:lgeq}
Models $(f_1)=(H_1,\Psi, f_1)$ and  $(f_2)=(H_2,\Psi, f_2)$ are
called $LG$-equivalent, if for each  $X\in \Gamma$ and $T\subset
\Phi(X)$ the following equality takes place
$$
T_{(f_1)}^{LL} = T_{(f_2)}^{LL}.
$$
\end{defin}

Recall that the logical kernel $LKer(\mu)$ of a point $\mu\in
Hom(W(X),H)$ is  $X$-$LG$-type of $\mu$. Denote by $S^{X}(f)$ the
set of all $X$-$LG$-types of the model $(f)$.

\begin{defin}
Models $(f_1)=(H_1, \Psi, f_1)$ and  $(f_2)=(H_2, \Psi, f_2)$ are
called $LG$-isotypic, if
$$
 S^{X}(f_1)=S^{X}(f_2),
$$
for each finite $X\in \Gamma$.
\end{defin}

In other words, models $(H_1,\Psi, f_1)$ and $(H_2,\Psi, f_2)$ are
$LG$-isotypic, if the subject areas of algebras $H_1$ and $H_2$ have
equal possibilities with respect to solution of logical formulas
from $T\subset \Phi(X)$ for each finite $X\in\Gamma$. The notions
of $LG$-isotypeness and $LG$-equivalence are tightly connected.

\begin{theorem}[\cite{APP}]\label{th:aa}
Models $(H_1, \Psi, f_1)$ and $(H_2, \Psi, f_2)$ are $LG$-equivalent
if and only if they are $LG$-isotypic.
\end{theorem}

\begin{remark}
Isotypiness of models imposes  some constraints on interpretations
$f_1$ and $f_2$. Let a point $\mu\in Hom(W(X),H_1)$ satisfies a
formula $u=\varphi(w_1,\ldots,w_m)$, $\varphi\in\Psi$. Then
$\nu\in Hom(W(X),H_2)$ satisfies the same $u$. Thus
$(w_1^\mu,\ldots,w_m^\mu))\in f_1(\varphi)\subset H_1^m$ if and
only if $(w_1^\nu,\ldots,w_m^\nu)\in f_2(\varphi)\subset H_2^m$.
In particular, $w_i^\mu=w_j^\mu$ if and only if $w_i^\nu=w_j^\nu$.
\end{remark}

The next theorem ties together  isotypeness of  models and
isomorphism of knowledge bases.

\begin{theorem}\label{th:iso}
If models $(H_1, \Psi, f_1)$ and $(H_2, \Psi, f_2)$ are isotypic
then the corresponding knowledge bases are isomorphic.
\end{theorem}

\begin{proof}
By Theorem~\ref{th:aa}, isotypic models are $LG$-equaivalent.
Theorem~6.12 from \cite{APP} states that if the models $(H_1,
\Psi, f_1)$ and $(H_2, \Psi, f_2)$ are $LG$-equivalent  then the
categories  $LG_\Theta(f_1)$ and $LG_\Theta(f_2)$ are isomorphic.
From diagram~(\ref{d:2}) follows that the categories
$F_\Theta(f_1)$ and  $F_\Theta(f_2)$ are isomorphic. In fact, they
coincide and isomorphism $\alpha$ is the identity
isomorphism of the categories. Therefore, knowledge bases $KB_1 =
KB(H_1,\Psi,f_1)$ and $KB_2= KB(H_2,\Psi,f_2)$ are isomorphic.
\end{proof}

\begin{defin}\label{d:iso}
Knowledge bases $KB_1 = KB(H_1,\Psi,f_1)$ and $KB_2=
KB(H_2,\Psi,f_2)$ are called isotypic if the models
$(H_1,\Psi,f_1)$ and $(H_2,\Psi,f_2)$ are isotypic.
\end{defin}

According to Theorem~\ref{th:iso}, isotypic knowledge bases are
isomorphic. Theorem~\ref{th:iso} generalizes the theorem from
\cite{PP_AU}, which states that knowledge bases over finite
automorphic models are informationally equivalent. It also
generalizes  the result from \cite{APP} about informational
equivalence  of isotypic knowledge bases.

Let us treat the isomorphism problem for knowledge bases  from a slightly different angle.

\begin{defin}\label{def:isom}
Let $\varphi_1$ and $\varphi_2$ be functors from a category $C_1$
to a category $C_2$. We will say that an isomorphism of functors
$S:\varphi_1 \to \varphi_2$ is given, if for each morphism $\nu:
A\to B$ from $C_1$ the following commutative diagram takes place
$$
\CD
\varphi_1(A) @> S_{A} >> \varphi_2(A)\\
@V \varphi_1(\nu)  VV @VV \varphi_2(\nu) V\\
\varphi_1(B) @>S_{B} >> \varphi_2(B).
\endCD
$$
Here $S_A$ is $A$-component of $S$, i.e., $S_A$ is a function
which provides a bijection between  $\varphi_1(A)$ and
$\varphi_2(A)$. The same condition holds true for $S_B$.
\end{defin}

An invertible functor from a category to itself is called \emph{an
automorphism} of a category.  An automorphism $\varphi$ of a
category  $C$ is called \emph{inner} (see \cite{Pl-St}) if $\varphi $ is
isomorphic to the identity functor $1_C$.

For each model $(H,\Psi,f)$ the correspondence $s\to s_\ast$
gives rise to functors
$$Cl_H:\Theta^0\to PoSet,
$$
$$
Cl_{(f)}:\widetilde \Phi\to Lat,
$$
where $PoSet$ is the category of partially ordered sets, $Lat$ is
the category of lattices. The functor $Cl_H$ assigns  a partially ordered set $C_\Theta^X(H)$ of all $H$-closed
congruences on  $W(X)$  to each
$W(X)$, while  $Cl_{(f)}$ assigns  a lattice of $H$-closed filters in $\Phi(X)$ to each
$\Phi(X)$.

Let us consider the commutative diagram
$$
\CD
\Theta^{0} @[2]> \varphi  >> \Theta^{0}\\
 @[2]/SE/ Cl_{H_1} //@.@.\;    @/SW// Cl_{H_2} /\\
 @. PoSet
 \endCD
 $$
where $\varphi$ is an automorphism of the category $\Theta^0$.
Commutativity of this diagram means that there is an isomorphism
of functors
$$
\alpha_\varphi:Cl_{H_1} \to Cl_{H_2} \cdot \varphi.
$$

This  isomorphism of functors means that the following
diagram is commutative
 $$
\CD
Cl_{H_1}(W(Y)) @> (\alpha_\varphi)_{W(Y)} >> Cl_{H_2}(\varphi(W(Y)) \\
 @V Cl_{H_1}(s) VV @VV Cl_{H_2}(\varphi(s)) V\\
 Cl_{H_1}(W(X)) @> (\alpha_\varphi)_{W(X)} >>
 Cl_{H_2}(\varphi(W(X)).
\endCD
$$
Similarly, the commutative diagram
$$
\CD \widetilde\Phi @[2]>  \varphi >> \widetilde\Phi \\
@[2]/SE/ Cl_{(f_1)} // @.@. \; @/SW//  Cl_{(f_2)} / \\
@. Lat
 \endCD
$$
gives rise to the isomorphism of functors
$$
\alpha_\varphi : Cl_{(f_1)} \to Cl_{(f_2)} \varphi.
$$


\begin{defin}[\cite{PP_D}]
Algebras $H_1$ and $H_2$ from a variety $\Theta$ are called
geometrically automorphically equivalent if for some automorphism
$\varphi$ of the category $\Theta^0$ there is the functor
isomorphism $\alpha_\varphi:Cl_{H_1} \to Cl_{H_2} \cdot \varphi$.
\end{defin}

\begin{defin}[\cite{PP_D}]
Models  $(H_1, \Psi, f_1)$ and  $(H_2, \Psi, f_2),$ where $H_1, H_2\in \Theta,$ are called logically
automorphically equivalent if for some automorphism $\varphi$ of
the category $\widetilde \Phi$ there is the functor isomorphism
$\alpha_\varphi :Cl_{(f_1)} \to Cl_{(f_2)} \cdot \varphi$.
\end{defin}

In the case of geometry over algebras, the following theorem is valid.

\begin{theorem}[\cite{PP_D}]\label{th:gqqsim}
Let $Var(H_1)=Var(H_2)=\Theta.$ If algebras  $H_1$ and
$H_2$ are geometrically automorphically equivalent then the
categories   $AG_\Theta(H_1)$ and $AG_\Theta(H_2)$ are isomorphic.
\end{theorem}

A   generalization of this result for the case of
logical geometry and models is of great interest. Here is the corresponding result (for
the proof see \cite{APP1}).

\begin{theorem}\label{th:gqqsimLL}
Let $(H_1,\Psi, f_1)$  and $(H_2,\Psi, f_2)$ be logically
automorphically equivalent models such that
$Var(H_1)=Var(H_2)=\Theta.$ Then the categories $LG_\Theta(f_1)$
and $LG_\Theta(f_2)$, and the corresponding knowledge bases
$KB_1=KB(H_1,\Psi, f_1)$ and  $KB_2 =KB(H_2,\Psi, f_2)$ are
isomorphic.
\end{theorem}

The following problem arises in a natural way.

\begin{problem}
Find necessary and sufficient conditions on  models $(H_1,
\Psi, f_1)$ and $(H_2, \Psi, f_2)$, which provide an isomorphism
of the corresponding knowledge bases.
\end{problem}

Theorem~\ref{th:iso} gives a sufficient condition for knowledge
bases isomorphism.

The following proposition plays an important role.

\begin{prop}[\cite{Pl-Sib}]\label{prop:eq}
Assume that for a variety $\Theta$ each automorphism of the
category  $\Theta^0$ is inner. The categories of algebraic sets
$AG_\Theta({H_1})$ and $AG_\Theta({H_2})$, where $H_1, \ H_2 \in
\Theta$,  are isomorphic if and only if the algebras $H_1$, $H_2$
are $AG$-equivalent.
\end{prop}

For knowledge bases Proposition~\ref{prop:eq} gives necessary and
sufficient conditions for knowledge bases isomorphism when
the set of relation symbols $\Psi$ of the corresponding knowledge
base $KB(H,\Psi, f)$ contains only equality predicate symbol. The
general case is still  open problem.

\begin{problem}
Let  models  $(H_1,\Psi,f_1)$ and $(H_2,\Psi,f_2)$ be given,
$H_1, H_2 \in \Theta$. Assume that for the variety  $\Theta$ each
automorphism of the category $\widetilde \Phi$ is inner. Is it
true that $LG_\Theta({f_1})$ and $LG_\Theta({f_2})$ are isomorphic
if and only if  $H_1$ and $H_2$ are $LG$-equivalent?
\end{problem}

Of course, the necessary and sufficient conditions depend on the
variety $\Theta$. In this respect it is interesting to consider
the following problem.

\begin{problem}
What are automorphisms of the category  $\widetilde\Phi$ for
various varieties $\Theta$.
\end{problem}

Note that for applications the varieties of groups and semigroups are of  special interest. We finish our discussion with the question, which is  also important for applications.

\begin{problem}\label{con:k}
What are necessary and sufficient conditions providing an isomorphism of finite knowledge bases.
\end{problem}

\noindent{\bf Acknowledgements.} E.~Aladova was  supported  by the
Israel Science Foundation grant No.~1207/12, by the Minerva
foundation through the Emmy Noether Research Institute. The
support of these institutions is gratefully appreciated.


\begin{thebibliography}{99}


\bibitem{APP}
 E. Aladova, E.Plotkin, T. Plotkin, \textit{Isotypeness of models and knowledge bases equivalence}, Mathematics in Computer Science, 7(4) (2013), 421--438.

\bibitem{APP1}
 E. Aladova, B.Plotkin, T. Plotkin, \textit{Similarity of models and knowledge bases isomorphism}, Preprint.



\bibitem{Gv}
A.~A.~Gvaramiya,  {\it Halmos algebras and axiomatizable classes of quasigroups}
 (Russian) Uspekhi Mat. Nauk {\bf 40} (1985), no. 4(244), 215--216.


\bibitem{Halmos}
P.R.~Halmos, {\it Algebraic logic}, New York, (1969).

\bibitem{MacLane}
S.~Mac Lane, {\it Categories for the Working Mathematician},
Graduate Texts in Mathematics, {\bf 5}, Springer-Verlag, New
York-Berlin, (1971). 

\bibitem{Marker}
D.~Marker, {\it Model Theory: An Introduction}, Springer Verlag,
(2002).

\bibitem{MPP1}
G.~Mashevitzky, B.~Plotkin, E.~Plotkin, Automorphisms of
categories of free algebras of varieties, {\it Electronic Research
Announcements of AMS}, {\bf 8} (2002) 1--10.

\bibitem{P9}
B.~Plotkin, {\it Universal algebra, algebraic logic and
databases.} Kluwer Acad. Publ., 1994.

\bibitem{Seven} B.Plotkin, Seven lectures on the universal algebraic geometry,
Preprint,(2002),  Arxiv:math, GM/0204245, 87pp.

\bibitem{Pl-St}
B. Plotkin, \textit{Algebras with the same algebraic geometry},
Proceedings of the Steklov Institute of Mathematics, MIAN, \textbf{242} (2003), 176--207.


\bibitem{Pl-Sib}
B. Plotkin.  Varieties of algebras and algebraic varieties. Categories of
       algebraic varieties, {\it Siberian Advanced Mathematics, Allerton Press}, {\bf 7:2} (1997) 64--97.



\bibitem{Plotkin_AG}
B.~Plotkin, {\it Algebraic geometry in First Order Logic},
Sovremennaja Matematika and Applications {\bf 22} (2004),
p.~16--62. Journal of Math. Sciences, {\bf 137}, n.5, (2006),
p.~5049-- 5097. http:// arxiv.org/ abs/ math GM/0312485.


\bibitem{Pl_GAGTA}
B. Plotkin, \textit{Algebraic logic and logical geometry in arbitrary varieties of algebras}, In: Proceedings of the Conf. on Group Theory, Combinatorics and Computing, AMS Contemporary Math series, \textbf{611}  (2014), 151--167.

\bibitem{PAP}
B. Plotkin, E. Aladova, E. Plotkin, \textit{Algebraic logic and logically-geometric types in varieties of algebras.} Journal of Algebra and its Applications, \textbf{12(2)} Paper No. 1250146, 23 p. (2013).





\bibitem{PP_D}
B. Plotkin, E. Plotkin, \textit{Multi-sorted logic and logical geometry: some problems.} Demonstratio Mathematica, to appear.


\bibitem{PP_AU}
B. Plotkin, T. Plotkin, \textit{Geometrical aspect of databases
and knowledge bases.} Algebra Universalis \textbf{46} (2001),
131--161.


\bibitem{PP_LNCS}
B. Plotkin, T. Plotkin, \textit{Categories of elementary sets over
algebras and categories of elementary algebraic knowledge.} LNCS,
Springer-Verlag, \textbf{4800} (2008), 555--570.

\bibitem{Sm}
J.D.H.Smith,  An Introduction to Quasigroups and their Representations. Chapman and Hall/CRC Press.  (2007).



\bibitem{Zhitom_types}
G.~Zhitomirski, On logically-geometric types of algebras, preprint
arXiv: 1202.5417v1 [math.LO].

\end{thebibliography}
\end{document}